\documentclass[a4paper,11pt]{article}

\usepackage{times}
\usepackage{fullpage}
\usepackage{amsfonts}
\usepackage{latexsym}
\usepackage{url}
\usepackage{amssymb,amsmath,amsthm}
\usepackage{color}
\usepackage{graphicx}
\usepackage{microtype}

\begin{document}

\bibliographystyle{plain}

\def\eps{\varepsilon}
\let\geq\geqslant
\let\leq\leqslant

\def\IR{\mathbb{R}}
\def\IZ{\mathbb{Z}}
\def\cell{\mathcal{C}}
\def\quadtree{\mathcal{T}}
\def\clist{\mathcal{L}}

\newtheorem{theorem}{Theorem}
\newtheorem{lemma}[theorem]{Lemma}
\newtheorem{proposition}[theorem]{Proposition}
\newtheorem{corollary}[theorem]{Corollary}

\makeatletter
\partopsep\z@ \textfloatsep 10pt plus 1pt minus 4pt
\def\section{\@startsection {section}{1}{\z@}{-3.5ex plus -1ex minus
-.2ex}{2.3ex plus .2ex}{\large\bf}}
\def\subsection{\@startsection{subsection}{2}{\z@}{-3.25ex plus -1ex
minus -.2ex}{1.5ex plus .2ex}{\normalsize\bf}}
\def\@fnsymbol#1{\ensuremath{\ifcase#1\or *\or 1\or 2\or
    3\or 4\or 5\or 6\or 7 \or 8\ or 9 \or 10\or 11 \else\@ctrerr\fi}}
\makeatother


\title{Reverse nearest neighbor queries in fixed dimension%
  \thanks{O.C.~and J.Y.~were supported by Mid-career Researcher
    Program through NRF grant funded by the~MEST
    (No.~R01-2008-000-11607-0). 
    The cooperation by A.V.~and J.Y.~was supported by the INRIA
    \emph{\'Equipe Associ\'ee}~KI.}}

\author{Otfried Cheong%
\thanks{Dept.\ of Computer Science, Korea Advanced Institute of
  Science and Technology, Gwahangno 335, Daejeon 305-701, South Korea.  
  Email: \textrm{otfried@kaist.edu}, \textrm{yeon-ju-young@kaist.ac.kr}.}
\and Antoine Vigneron%
\thanks{INRA, UR 341 Math\'ematiques
  et Informatique Appliqu\'ees,
  78352 Jouy-en-Josas, France.
  Email: \mbox{\textrm{antoine.vigneron@jouy.inra.fr}}.}
\and Juyoung Yon\footnotemark[2]}

\date{\today}

\maketitle

\begin{abstract}
  Reverse nearest neighbor queries are defined as follows: Given an
  input point set $P$, and a query point $q$, find all the points $p$
  in $P$ whose nearest point in $P \cup \{q\} \setminus \{p\}$ is $q$.
  We give a data structure to answer reverse nearest neighbor queries
  in fixed-dimensional Euclidean space.  Our data structure uses
  $O(n)$ space, its preprocessing time is $O(n \log n)$, and its query
  time is $O(\log n)$.
\end{abstract}


\section{Introduction}
Given a set $P$ of $n$ points in $\IR^d$, a well-known problem in
computational geometry is nearest neighbor searching: 
preprocess $P$ such that, for any query point $q$, a point in $P$
that is closest to $q$ can be reported efficiently. This problem
has been studied extensively; in this paper, we consider the
related problem of {\em reverse nearest neighbor searching},
which has attracted some attention recently.

The reverse nearest neighbor searching problem is the following. 
Given a query point $q$, we want to report all the points in $P$
that have $q$ as one of their nearest neighbors. More
formally, we want to find the points $p \in P$ such that for all
points $p' \in P \setminus \{p\}$, the distance $|pp'|$ is larger or
equal to the distance~$|pq|$.

The earliest work on reverse nearest neighbor searching is by
Korn and Muthukrishnan~\cite{Korn00}. 
They motivate this problem by applications in databases.
Their approach is based on R-Trees, so
it is unlikely to give a good worst-case time bound. Subsequently,
the reverse nearest neighbor searching problem has attracted
some attention in the database
community~\cite{Benetis06,KumarJG08,Lin08,Stanoi00,Tao07,Xia06,Yang01,Yiu07}.

The only previous work on reverse nearest neighbor searching
where worst-case time bounds are given is the work by 
Maheshwari et al.~\cite{Maheshwari02}. They give 
a data structure for the two-dimensional case, using $O(n)$ space,
with $O(n\log n)$ preprocessing time, and $O(\log n)$ query time. 
Their approach is to show that the arrangement of the largest empty
circles centered at data points has linear size, and then they
answer queries by doing point location in this arrangement.

In this paper, we extend the result of Maheshwari et
al.~\cite{Maheshwari02} to arbitrary fixed dimension.  We give a data
structure for reverse nearest neighbor searching in $\IR^d$, where
$d=O(1)$, using the Euclidean distance. Our data structure has size
$O(n)$, with preprocessing time $O(n \log n)$, and with query time
$O(\log n)$.
It is perhaps surprising that we can match the bounds for the
two-dimensional case in arbitrary fixed dimension.  For \emph{nearest
neighbor queries}, this does not seem to be possible: The bounds for
nearest neighbor searching in higher dimension depend on the
complexity of the Voronoi diagram, which is $\Theta(n^{\lceil d/2
\rceil})$ in $d$-dimensional space.  

Our approach is similar to some previous work on approximate Voronoi
diagrams~\cite{AMM09,Harpeled,HP01}: the space is partitioned using
a compressed quadtree, each cell of this quadtree containing a small 
set of candidate points. Queries are answered by finding the cell
containing the query point, and checking all the candidate points in
this cell. Interestingly, this approach allows to answer reverse
nearest neighbor queries efficiently and {\em exactly}, while it only
seems to give approximations for nearest neighbor searching.

Our model of computation is the real-RAM model, with some additional
operations that are common in quadtree algorithms, such as the floor
function, the logarithm $\log_2$, and the bitwise \textsc{xor}. In particular,
we need to be able to find in constant time the first binary digit
at which two numbers differ.  This allows, for instance, to find in constant
time the smallest quadtree box that contains two given points. For more
details on this issue, we refer to the lecture notes of 
Har-Peled~\cite{Harpeled}, and to previous work related to 
quadtrees~\cite{Bern99,EppsteinGS08}.

\section{Compressed quadtrees}

In this section, we describe compressed quadtrees, a well known 
data structure in computational geometry. A more detailed presentation
can be found in Har-Peled's lecture notes~\cite{Harpeled}, in the article on
skip quadtrees by Eppstein, Goodrich, and Sun~\cite{EppsteinGS08}, or in the
article by Bern, Eppstein and Teng~\cite{Bern99}.
We first describe quadtrees, and then we describe their compressed
version.

We consider quadtrees in $\IR^d$, where $d=O(1)$. We denote by $H_r$
the hypercube $[-1,1]^d$; the leaves of a quadtree will form a
partition of $H_r$.

A {\em quadtree box} is either $H_r$, or is obtained by partitioning
a quadtree box $H$ into $2^d$ equal sized hypercubes---these hypercubes
are called the {\em quadrants} of $H$.
A quadtree is a data structure that stores 
quadtree boxes in a hierarchical manner.
Each node $\nu$ of a quadtree stores a quadtree box $\cell(\nu)$,
and pointers to its parent and its children.
We call $\cell(\nu)$ the {\em cell} of node $\nu$.
In this paper, the cell of the root of a quadtree is always the box $H_r$.
Each node $\nu$  either is a leaf, or has
$2^d$ children that store the $2^d$ quadrants of~$\cell(\nu)$.
With this definition, the cells of the leaves of a quadtree
form a partition of $H_r$.

Let $S$ denote a set of $m$ quadtree boxes.  We can construct
the smallest quadtree whose nodes store all boxes in $S$ as follows.
We start by constructing the root. If $S \subset \{H_r\}$, 
then we are done. Otherwise, we construct the $2^d$ children
of the root. We consider the subset $S_1 \subset S$ (resp. $S_2,S_3,\dots$)
of the boxes in $S$ contained in the first quadrant (resp. second, third,\dots).
We construct recursively the quadtree, by taking the first
(resp. second, third, \dots) 
child as the root and using the set of boxes $S_1$ (resp. $S_2, S_3,\dots$).

The above construction results in a quadtree that stores all the boxes
in $S$.  Even though it is the smallest such quadtree, its size can be
arbitrarily large when $S$ contains very small boxes. To remedy this,
we use a \emph{compressed quadtree}, which allows to bypass long
chains of internal nodes.

In order to reduce the size of the data structure, we allow two
different kinds of nodes in a compressed quadtree.  An \emph{ordinary
node} stores a quadtree box as before.  A \emph{compressed node}
$\nu$, however, stores the difference $H \setminus H'$ of two quadtree
boxes $H$ and $H'$.  We still call this difference the
cell~$\cell(\nu)$.  Compressed nodes are always leaves of the
compressed quadtree.

As in a quadtree, the cells of the children of a node $\nu$
form a partition of $\cell(\nu)$. But two cases are now possible:
either these cells are the quadrants of $\cell(\nu)$, or $\nu$ has two
children, one of them storing a quadtree box $H \subset \cell(\nu)$,
and the other storing $\cell(\nu) \setminus H$.

The construction of a compressed quadtree that stores all the boxes
in $S$ is analogous to the construction of the ordinary quadtree,
with the following difference. Assume we are processing an internal
node $\nu$. Let $H$ denote the smallest quadtree box
containing the boxes in $S$ that are strictly contained in $\cell(\nu)$.
If $H=\cell(\nu)$, then we proceed exactly as we did for the ordinary
quadtree: we construct $2^d$ children corresponding to the quadrants
of $\cell(\nu)$. Otherwise, $\nu$ has two children, one stores $H$, and the 
other is a  compressed node that stores $\cell(\nu) \setminus H$. 
Intuitively, this construction of a compressed node allows us to
``zoom in'' when all the boxes in $\cell(\nu)$ are within a small area,
and avoids a long chain of internal nodes. (See Figure~\ref{fig:quadtree}.)

\begin{figure}[ht]
\centerline{\includegraphics{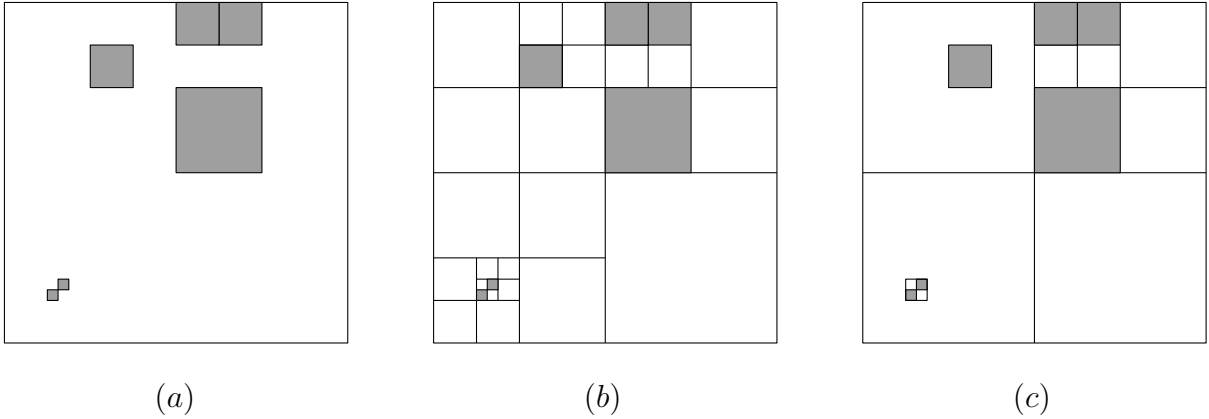}}
\caption{(a) A set of quadtree boxes. (b) The quadtree storing
  these quadtree boxes. (c) The compressed quadtree storing
  the same set of quadtree boxes. The two cells on the left side
  correspond to compressed nodes.\label{fig:quadtree}}
\end{figure}


A direct implementation of the quadtree construction that we just
described would lead to $O(m^2)$ construction time and $O(m)$ query
time, since a quadtree may have linear depth. However, it is possible
to achieve $O(m \log m)$ construction time and $O(\log m)$ query time
using different algorithms for constructing and querying a compressed
quadtree. One such construction is presented in Har-Peled's lecture
notes~\cite{Harpeled}. The idea is first to find a quadtree box
$\cell_0$ that contains a constant fraction of the input boxes~$S$,
which can be done in linear time by choosing the center among the
vertices of an irregular, constant-size grid.  Then one computes
recursively a compressed quadtree for the set $S_{in}$ of the boxes
in~$S$ that are contained in~$\cell_0$, and for $S_{out}=S \setminus
S_{in}$. Finally, these two quadtrees are merged in linear time,
essentially by hanging the quadtree of $S_{in}$ at a leaf
of~$S_{out}$.

The quadtree can be queried in $O(\log m)$ time by constructing a
\emph{finger tree} over the quadtree, which is an auxiliary data
structure to allow faster search. (This approach is also presented in
Har Peled's notes~\cite{Harpeled}.) We first find a cell of the
quadtree such that the subtree rooted at the corresponding node
contains a constant fraction of the boxes in~$S$. This node is called
a \emph{separator} and can be found in linear time by traversing the
quadtree. This construction is repeated recursively on the forest
obtained by removing this node. The construction allows to answer a
query in $O(\log m)$ time, as this tree has $O(\log m)$ depth. So we
have the following bounds for constructing and querying a quadtree:
\begin{lemma}\label{lem:quadtree}
  Let $S$ be a set of $m$ quadtree boxes contained in $H_r$. We can
  construct in time $O(m \log m)$ a compressed quadtree~$\quadtree$,
  such that each box in $S$ is the cell of a node 
  of~$\quadtree$. This compressed quadtree~$\quadtree$ has
  size~$O(m)$. After $O(m \log m)$ preprocessing time, we can find for
  any query point $q \in H_r$ the leaf of~$\quadtree$ whose cell 
  contains~$q$ in time $O(\log m)$.
\end{lemma}
Note that a query point might lie on the boundaries of several cells.
In this case, we break the tie arbitrarily, and we return only one cell
containing~$q$. 

\section{Data structure for reverse nearest neighbor queries}

In this section, we describe the construction of our data structure and
how we answer reverse nearest neighbor queries. This data structure
is a compressed quadtree, with a set of candidate points stored at each 
node. To answer a query, we locate the leaf of the compressed quadtree 
whose cell contains the query point, 
and we check all the candidate points in this leaf; the
reverse nearest-neighbors are among these points.
We start with some notation.

Our input point set is denoted by $P=\{p_1,\dots,p_n\}$, with $n \geq 2$. 
We still work in $\IR^d$, where $d=O(1)$, and so $P \subset \IR^d$.
The {\em empty ball} $b_i$ is the largest ball centered
at $p_i \in P$ that does not contain any other point of $P$ in its
interior. In other words, the boundary of the empty ball centered at $p$
goes through the nearest point to $p$ in $P \setminus \{p\}$. 
In this paper, we only consider closed balls, so 
$p_i$ is a reverse nearest neighbor of a query point $q$ if and
only if $q \in b_i$.

Let $H_P$ be a smallest axis-aligned $d$-dimensional hypercube containing
the input point set $P$. Without loss of generality, we assume
that $H_P=[-1/2\sqrt{d},1/2\sqrt{d}]^d$; then any empty ball is
contained in $H_r=[-1,1]^d$. When $\nu$ is an ordinary node,
we denote by $s(\nu)$ the side length of the quadtree box $\cell(\nu)$.

\begin{figure}
  \centerline{\includegraphics[scale=0.8]{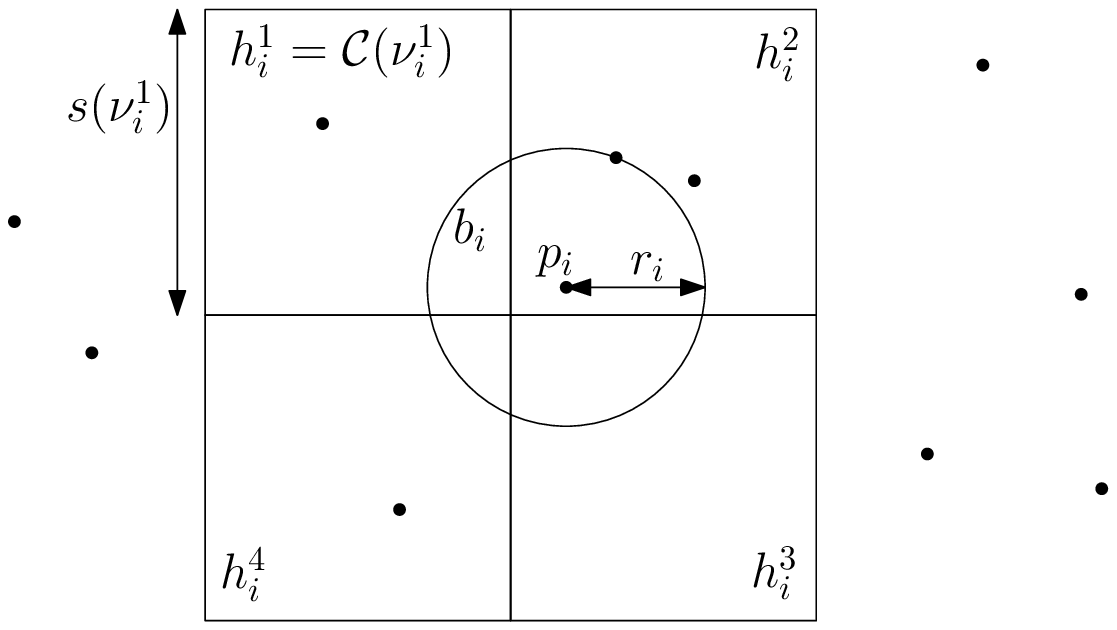}}
  \caption{Notation.\label{f:notation}}
\end{figure}

We first compute the set of all the largest empty balls for $P$. 
This can be done in $O(n\log n)$ time using
Vaidya's all-nearest neighbors algorithm~\cite{Vaidya89}.
We denote by $r_i$ the radius of $b_i$.
For each $b_i$, we compute the quadtree  boxes with side length 
in $[2r_i,4r_i)$ that overlap $b_i$. (See Figure~\ref{f:notation}.)
Our model of computation allows us to do this in $O(1)$ time.
There are at most $2^d$
such boxes; we denote them by $h_i^j, j \in \{1, \dots, 2^d\}$.

Using Lemma~\ref{lem:quadtree}, 
we construct in $O(n \log n)$ time a compressed quadtree $\quadtree$ 
of size $O(n)$  such that each box $h_i^j$ appears in $\quadtree$.
For each node $\nu$ of $\quadtree$, if the corresponding
cell $\cell(\nu)$ is $h_i^j$, we store $p_i$ as a candidate
point for $\nu$.  Storing these candidate points can be done during 
the construction of the quadtree within the same time bound. 
Notice that we may store several candidate
points for a given node $\nu$. 

These sets of candidate points are not sufficient for our purpose,
so we will add some other points. For each ordinary (non-compressed) node
$\nu$, we store the points $p_i$ such that $r_i>s(\nu)/4$ and $b_i$
overlaps $\cell(\nu)$; this list of candidate points is denoted 
by $\clist(\nu)$.  In order to analyze our algorithm, we need
the following lemma, which is proved in Section~\ref{sec:packing}.
\begin{lemma}\label{lem:size}
  For any ordinary node $\nu$, the cardinality of the set of candidate
  points $\clist(\nu)$ stored at~$\nu$ is~$O(1)$.
\end{lemma}

We construct the lists $\clist(\cdot)$ by traversing $\quadtree$ recursively,
starting from the root. Assume that $\nu$ is the current ordinary node.
The points $p_i$ such that $h_i^j=\cell(\nu)$ for some $j$ have already been
stored at $\nu$. By our construction,
they are the points $p_i$ in $\clist(\nu)$ such that 
$s(\nu)/4<r_i\leq s(\nu)/2$. So we need the other candidate points $p_k$,
such that $r_k > s(\nu)/2$. These points can be found in $\clist(\nu')$,
where $\nu'$ is the parent of $\nu$. So we insert in $\clist(\nu)$ 
all the points $p_k \in \clist(\nu')$ such that $b_k$ overlaps $\cell(\nu)$,
which completes the construction of $\clist(\nu)$. By Lemma~\ref{lem:size},
this can be done in $O(1)$ time per node, and thus overall, computing
the lists of candidate points for ordinary nodes takes $O(n)$ time.

If $\nu$ is a compressed node, and $\nu'$ is its parent, we just
set $\clist(\nu)=\clist(\nu')$. We complete the construction of our
data structure by handling all the compressed nodes.

Given a query point $q$, we answer reverse nearest-neighbor queries
as follows. If $q \notin H_r$, then we return $\emptyset$, because
we saw earlier that all empty balls are in $H_r$. Otherwise, we
find the leaf $\nu$ such that $q \in \cell(\nu)$, which can be done
in $O(\log n)$ time by Lemma~\ref{lem:quadtree}. For each point
$p_i \in \clist(\nu)$, we check whether $p_i$ is a reverse nearest
neighbor, that is, we check whether $q \in b_i$. If this is the
case, we report~$p_i$.

We still need to argue that we answered the query correctly. Assume
that $p_k$ is a reverse nearest neighbor of $q$, and the leaf $\nu$
containing $q$ is an ordinary node. As $q \in b_k$, we have
$q \in h_k^j$ for some $j$, and since the side length of $h_k^j$
is less than $4r_k$, we have $s(\nu) < 4r_k$. Since $b_k$ contains
$q$, it overlaps $\cell(\nu)$, so by definition of $\clist(\nu)$,
we have $p_k \in \clist(\nu)$, and thus $p_k$ was reported.
If $\nu$ is a compressed node, then the same proof works if we replace
$\nu$ by its parent $\nu'$, since $\clist(\nu)=\clist(\nu')$.

The discussion above proves the main result of this paper:
\begin{theorem}\label{th:main}
  Let $P$ be a set of $n$ points in $\IR^d$. We assume that
  $d=O(1)$. Then we can construct in time $O(n \log n)$ a data
  structure of size $O(n)$ that answers reverse nearest-neighbor
  queries in $O(\log n)$ time.  The number of reverse nearest
  neighbors is~$O(1)$.
\end{theorem}
The fact that the number of reverse nearest neighbors is $O(1)$ was
known before: In fixed dimension, the in-degree of the vertices of the
nearest neighbor graph is bounded by a constant.


\section{Proof of Lemma~\ref{lem:size}}\label{sec:packing}

In this section, we prove Lemma~\ref{lem:size}, which was needed
to establish the time bounds in Theorem~\ref{th:main}.
We start with a packing lemma.

\begin{lemma}\label{lem:packing}
Let $b$ be a ball with radius $r$. Then $b$ intersects
at most $2 \times 5^d$ empty balls with radius larger or equal
to $r$.
\end{lemma}
\begin{proof}
When $x,y \in \IR^d$, we denote by $|xy|$ the Euclidean distance between
$x$ and $y$, and we denote by $\overline{xy}$ the line segment connecting
them.

We denote by $c$ the center of $b$, and 
we denote by $b'$ the ball with center $c$ and radius $2r$.
We first bound the number of empty balls with radius $\geq r$ 
whose center is contained in $b'$. Let $B$ denote this set
of balls, and let $C$ denote the set
of their centers. Any two points in $C$ 
are at distance at least $r$ from each other. Hence, 
the balls with radius $r/2$ and with centers in $C$ are disjoint.
As they are all contained in the ball $b''$ with center $c$ and radius
$5r/2$, the sum of their volumes is at most the volume of $b''$.
Hence, we have $|C|\leq 5^d$, and thus $|B| \leq 5^d$.

We now consider the empty balls with radius $\geq r$ that intersect
$b$, and whose centers are not in $b'$. We denote by $B'$ the
set of these balls, and we denote by $C'$ the set of their
centers. Let $b_1$ (resp. $b_2$) be a ball in $B'$ with radius
$r_1$ (resp. $r_2$) and center $c_1$ (resp. $c_2$). 
(See Figure~\ref{fig:packing}.)
\begin{figure}[ht]
  \centerline{\includegraphics{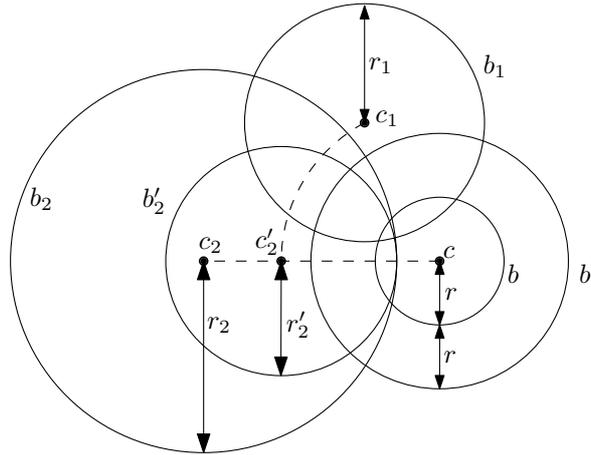}}
  \caption{Proof of Lemma~\ref{lem:packing}.\label{fig:packing}}
\end{figure}
Without loss of generality, we assume that $r_1 \leq r_2$.

Let $c_2'$ be the point of $\overline{cc_2}$ such that $|cc_2'|=|cc_1|$.
Let $r_2'=r_2-|c_2c_2'|$, and let $b_2'$ denote the ball with
center $c_2'$ and radius $r_2'$. As $b_2' \subset b_2$, we know
that $b_2'$ does not contain $c_1$ in its interior, 
and thus $r_2' \leq |c_1c_2'|$. As $b_2'$ intersects
$b$, we have $|cc_2'| \leq r+r_2'$. 
It implies that $|cc'_2|-r \leq |c_1c_2'|$.
Since $|cc_2'| \geq 2r$, it follows that $|cc_2'| \leq 2|c_1c_2'|$.

Let $c_1''$ (resp. $c_2''$) denote the projection of $c_1$ (resp. $c_2$)
onto the unit sphere $u$ centered at $c$. In other words, 
$c_1''=c+(1/|cc_1|)(c_1-c)$. Then it follows from the previous paragraph
that $|c_1''c_2''| \geq 1/2$. Hence, the spheres with radius $1/4$ and 
centered at the projections onto $u$ of the points in $C'$ are disjoint.
As these spheres are contained in the sphere with radius $5/4$ centered
at $c$, we have $|C'| \leq 5^d$, and thus $|B'| \leq 5^d$.
\end{proof}

Now we prove Lemma~\ref{lem:size}: For any ordinary node $\nu$, the number
of candidate points stored in $\clist(\nu)$ is $O(1)$. We assume
that $p_i \in \clist(\nu)$. By definition, we must have $r_i>s(\nu)/4$,
and $b_i$ overlaps $\cell(\nu)$. As $\cell(\nu)$ can be covered by
$O(1)$ balls with radius $s(\nu)/4$, Lemma~\ref{lem:packing} implies
that there can be only $O(1)$ such candidate points.

\section{Concluding remarks}

Our approach does not only give a data structure to answer reverse
nearest neighbor queries, it also yields a {\em reverse Voronoi
diagram}: a spatial subdivision with
linear complexity such that, within each cell, the set of reverse
nearest neighbors is fixed. To achieve this, we construct, within the 
cell of each leaf of our quadtree, the arrangement of
the empty balls of the candidate points. As there is only a constant
number of candidates per cell, each such arrangement has constant
complexity, so overall we get a subdivision of linear size.

The time bounds of our data structure can be improved in the word RAM
model of computation, when the coordinates of the input points are
$O(\log n)$-bits integers. In this case, Chan showed that 
the all-nearest neighbors computation and the compressed quadtree 
construction can be done in linear time, so our data structure can 
be built in linear time as well. Then using the shuffle-and-sort 
approach of Chan~\cite{Chan02}, combined with van Emde Boas trees,
the compressed quadtree in our data structure can be queried in
$O(\log \log n)$ time. So overall, we can construct in linear time
a data structure for reverse nearest neighbors with query time
$O(\log \log n)$.

The most natural extension to this problem would be to handle
different metrics. Our approach applies directly to any norm of
$\IR^d$, with $d=O(1)$, as its unit ball can be made \emph{fat} after
changing the coordinate system: we just need to apply an affine map
such that the John ellipsoid of the unit ball of this norm becomes a
Euclidean ball.  The time bounds and space usage remain the same.

Another possible extension would be to make our algorithm dynamic.
The main difficulty is that it seems that we would need to maintain the
empty balls, which means maintaining all nearest neighbors.
The result of Maheshwari et al.~\cite{Maheshwari02}, combined with
the data structure of Chan for dynamic nearest neighbors~\cite{Chan06},
gives polylogarithmic update time and query time in $\IR^2$.
In higher dimension, these bounds would be considerably worse, if
one uses the best known data structures for dynamic nearest 
neighbors~\cite{Agarwal95,Chan03}.

Finally, it would be interesting to find the dependency of our time
bounds on the dimension~$d$. We did not deal with this issue, because
one would first have to find this dependency for constructing
compressed quadtrees, which is not the focus of this paper.


\begin{thebibliography}{10}

\bibitem{Agarwal95}
P.~Agarwal and J.~Matousek.
\newblock Dynamic half-space range reporting and its applications.
\newblock {\em Algorithmica}, 13(4):325--345, 1995.

\bibitem{AMM09}
S.~Arya, T.~Malamatos, and D.~M. Mount.
\newblock Space-time tradeoffs for approximate nearest neighbor searching.
\newblock {\em Journal of the ACM}, 57:1--54, 2009.

\bibitem{Benetis06}
R.~Benetis, C.~Jensen, G.~Karciauskas, and S.~Saltenis.
\newblock Nearest and reverse nearest neighbor queries for moving objects.
\newblock {\em VLDB Journal}, 15(3):229--249, 2006.

\bibitem{Bern99}
M.~Bern, D.~Eppstein, and S.-H. Teng.
\newblock Parallel construction of quadtrees and quality triangulations.
\newblock {\em International Journal of Computational Geometry and
  Applicaptions}, 9(6):517--532, 1999.

\bibitem{Chan02}
T.~Chan.
\newblock Closest-point problems simplified on the {RAM}.
\newblock In {\em Proc. ACM-SIAM Symposium on Discrete Algorithms}, pages
  472--473, 2002.

\bibitem{Chan03}
T.~Chan.
\newblock Semi-online maintenance of geometric optima and measures.
\newblock {\em SIAM Journal on Computing}, 32(3):700--716, 2003.

\bibitem{Chan06}
T.~Chan.
\newblock A dynamic data structure for 3-d convex hulls and 2-d nearest
  neighbor queries.
\newblock In {\em Proc. ACM-SIAM Symposium on Discrete Algorithms}, pages
  1196--1202, 2006.

\bibitem{EppsteinGS08}
D.~Eppstein, M.~Goodrich, and J.~Sun.
\newblock Skip quadtrees: Dynamic data structures for multidimensional point
  sets.
\newblock {\em International Journal of Computational Geometry and
  Applications}, 18(1/2):131--160, 2008.

\bibitem{Harpeled}
S.~Har-Peled.
\newblock Geometric approximation algorithms.
\newblock Lecture notes, available on the author's webpage.

\bibitem{HP01}
S.~Har-Peled.
\newblock A replacement for {Voronoi} diagrams of near linear size.
\newblock In {\em Proc. Symposium on Foundations of Computer Science}, pages
  94--103, 2001.

\bibitem{Korn00}
F.~Korn and S.~Muthukrishnan.
\newblock Influence sets based on reverse nearest neighbor queries.
\newblock In {\em Proc. ACM SIGMOD International Conference on Management of
  Data}, pages 201--212. ACM, 2000.

\bibitem{KumarJG08}
Y.~Kumar, R.~Janardan, and P.~Gupta.
\newblock Efficient algorithms for reverse proximity query problems.
\newblock In {\em 16th ACM SIGSPATIAL International Symposium on Advances in
  Geographic Information Systems}, page~39, 2008.

\bibitem{Lin08}
J.~Lin, D.~Etter, and D.~DeBarr.
\newblock Exact and approximate reverse nearest neighbor search for multimedia
  data.
\newblock In {\em Proc. SIAM International Conference on Data Mining}, pages
  656--667, 2008.

\bibitem{Maheshwari02}
A.~Maheshwari, J.~Vahrenhold, and N.~Zeh.
\newblock On reverse nearest neighbor queries.
\newblock In {\em Proc. Canadian Conference on Computational Geometry}, pages
  128--132, 2002.

\bibitem{Stanoi00}
I.~Stanoi, D.~Agrawal, and A.~El Abbadi.
\newblock Reverse nearest neighbor queries for dynamic databases.
\newblock In {\em Proc. ACM SIGMOD Workshop on Research Issues in Data Mining
  and Knowledge Discovery}, pages 44--53, 2000.

\bibitem{Tao07}
Y.~Tao, D.~Papadias, X.~Lian, and X.~Xiao.
\newblock Multidimensional reverse {kNN} search.
\newblock {\em VLDB Journal}, 16(3):293--316, 2007.

\bibitem{Vaidya89}
P.~Vaidya.
\newblock An {$O(n \log n)$} algorithm for the all-nearest-neighbors problem.
\newblock {\em Discrete {\&} Computational Geometry}, 4:101--115, 1989.

\bibitem{Xia06}
T.~Xia and D.~Zhang.
\newblock Continuous reverse nearest neighbor monitoring.
\newblock In {\em Proc. International Conference on Data Engineering}, page~77,
  2006.

\bibitem{Yang01}
C.~Yang and K.~Lin.
\newblock An index structure for efficient reverse nearest neighbor queries.
\newblock In {\em Proc. International Conference on Data Engineering}, pages
  485--492, 2001.

\bibitem{Yiu07}
M.~Yiu and N.~Mamoulis.
\newblock Reverse nearest neighbors search in ad hoc subspaces.
\newblock {\em IEEE Transactions on Knowledge and Data Engineering},
  19(3):412--426, 2007.

\end{thebibliography}

\end{document}